\documentclass{article}

\usepackage[light,easyscsl,sfmathbb,nowarning]{kpfonts}
\usepackage{fullpage}
\usepackage{url}
\usepackage[hidelinks]{hyperref}
\usepackage[utf8]{inputenc}
\usepackage[small]{caption}
\usepackage{graphicx}
\usepackage{amsmath}
\usepackage{amsthm}
\usepackage{amsfonts}
\usepackage{booktabs}
\usepackage{amssymb,mathtools}
\usepackage{bm}
\usepackage[pdftex,dvipsnames]{xcolor}
\usepackage{authblk}
\usepackage{xspace}
\usepackage{comment}
\usepackage{array}
\usepackage{multirow}
\usepackage{nicefrac}
\usepackage{todonotes}
\usepackage[noend]{algpseudocode}
\usepackage[numbers,sort]{natbib}

\usepackage[british]{babel}
\usepackage{hyperref}
\hypersetup{
    colorlinks=true,
    linkcolor=olive,
    citecolor=purple
}

\usepackage{booktabs}
\usepackage[ruled,linesnumbered]{algorithm2e}
\usepackage{diagbox}

\newtheorem{theorem}{Theorem}[section]
\newtheorem{claim}[theorem]{Claim}

\newtheorem{proposition}[theorem]{Proposition}
\newtheorem{lemma}[theorem]{Lemma}
\newtheorem{corollary}[theorem]{Corollary}

\newcommand{\ProblemFormat}[1]{{\sc #1}}
\newcommand{\ProblemName}[1]{\ProblemFormat{#1}\xspace}

\newcommand{\schulzew}[0]{\ProblemName{Schulze-Winner}}
\newcommand{\schulzeaw}[0]{\ProblemName{Schulze-AllWinners}}
\newcommand{\schulzewd}[0]{\ProblemName{Schulze-WinnerDetermination}}

\newcommand{\Oh}{\ensuremath{\mathcal{O}}}

\newcommand{\tO}{\tilde{\Oh}}

\newcommand{\calW}{\ensuremath{\mathcal{W}}}

\newcommand{\naturals}{{{\mathbb{N}}}}

\newcommand{\eps}{\varepsilon}

\hypersetup{pdfinfo={
  Title={Fine-Grained Complexity and Algorithms for the Schulze Voting Method},
  Author={Krzysztof Sornat, Virginia Vassilevska Williams and Yinzhan Xu},
  Keywords={voting, Schulze method, algorithms, fine-grained complexity}
}}

\title{Fine-Grained Complexity and Algorithms for\\
       the Schulze Voting Method\thanks{A previous version of this work appears in EC 2021~\cite{SornatVWX21}. In this version we strengthen Theorem~\ref{thm:lb_testing_winner} which now holds also for the problem of finding a Schulze winner.}}

\date{}
\author[1]{Krzysztof Sornat\thanks{sornat@mit.edu. Now at AGH University, Poland.}}
\author[1]{Virginia Vassilevska Williams\thanks{virgi@mit.edu}}
\author[1]{Yinzhan Xu\thanks{xyzhan@mit.edu}}
\affil[1]{MIT CSAIL, USA}

\begin{document}

\maketitle

\begin{abstract}

We study computational aspects of a well-known single-winner voting rule called the Schulze method [Schulze, 2003] which is used broadly in practice. In this method the voters give (weak) ordinal preference ballots which are used to define the weighted majority graph of direct comparisons between pairs of candidates. The choice of the winner comes from indirect comparisons in the graph, and more specifically from considering directed paths instead of direct comparisons between candidates.

When the input is the weighted majority graph, to our knowledge, the fastest algorithm for computing all winners in the Schulze method uses a folklore reduction to the All-Pairs Bottleneck Paths (APBP) problem and runs in $\Oh(m^{2.69})$ time, where $m$ is the number of candidates. It is an interesting open question whether this can be improved. Our first result is a combinatorial algorithm with a nearly quadratic running time for computing all winners. This running time is essentially optimal as it is nearly linear in the size of the weighted majority graph.

If the input to the Schulze winners problem is not the weighted majority graph but the preference profile, then constructing the weighted majority graph is a bottleneck that increases the running time significantly; in the special case when there are $m$ candidates and $n=\Oh(m)$ voters, the running time is $\Oh(m^{2.69})$, or $\Oh(m^{2.5})$ if there is a nearly-linear time algorithm for multiplying dense square matrices.

To address this bottleneck, we prove a formal equivalence between the well-studied Dominance Product problem and the problem of computing the weighted majority graph. As the Dominance Product problem is believed to require at least time $r^{2.5-o(1)}$ on $r\times r$ matrices, our equivalence implies that constructing the weighted majority graph in $\Oh(m^{2.499})$ time for $m$ candidates and $n = \Oh(m)$ voters would imply a breakthrough in the study of ``intermediate'' problems [Lincoln et al., 2020] in fine-grained complexity. We prove a similar connection between the so called Dominating Pairs problem and the problem of finding a winner in the Schulze method.

Our paper is the first to bring fine-grained complexity into the field of computational social choice. Previous approaches say nothing about lower bounds for problems that already have polynomial time algorithms. By bringing fine-grained complexity into the picture we can identify voting protocols that are unlikely to be practical for large numbers of candidates and/or voters, as their complexity is likely, say at least cubic.

\end{abstract}

\newpage
\section{Introduction}\label{sec:intro}
We study computational aspects of {\it the Schulze method}\footnote{The method appears under different names such as {\it Beatpath Method/Winner}, {\it Path Voting}, {\it (Cloneproof) Schwartz Sequential Dropping}---see discussion in~\cite{Schulze03}.},
a single-winner voting rule defined on (weak) ordinal preference ballots.
The method was introduced by Markus Schulze~\cite{Schulze03,Schulze11} and has been extensively researched over the years---see, e.g., a survey by Schulze~\cite{Schulze-arxiv18}.

The Schulze voting method definition uses a representation of the votes called {\it the weighted majority graph}, a graph whose vertices are the candidates and whose directed weighted edges roughly capture all pairwise comparisons between candidates (see a formal definition in Section~\ref{sec:preliminaries}).
The Schulze method computes, in the weighted majority graph, for every pair of candidates $u$ and $v$ the smallest weight $B(u,v)$ of an edge on {\it a widest path} between $u$ and $v$.
Here a widest path is a $u$-$v$ path whose minimum edge weight is maximized over all $u$-$v$ paths.
Then, a {\it Schulze winner} is any candidate $u$ s.t. for all other candidates $v$, $B(u,v)\geq B(v,u)$.
A Schulze winner always exists~\cite{Schulze11}, though it may not be unique.
The method can be used to provide a (weak) order over alternatives, and hence can be seen as a preference aggregation method or multi-winner voting rule (taking the top candidates in the ranking as winners).

Schulze defined this method by modifying {\it the Minimax Condorcet method}\footnote{
Also referred to as {\it the Minimax method} and {\it the Simpson-Kramer method}~\cite{Schulze11}.}
to address several criticisms of the Minimax Condorcet method~\cite{Smith73,Tideman87,Saari94}.
Schulze's modified method is {\it Condorcet consistent}, i.e., a candidate who is preferred by a majority over every other candidate in pairwise comparisons (a Condorcet winner) is a winner in the Schulze method.
Schulze's modification results in satisfying a desired set of axiomatic properties not satisfied by Minimax Condorcet:
\begin{itemize}
    \item {\it clone independence:} This criterion proposed by Tideman~\cite{Tideman87}
    requires that a single-winner voting rule should be independent of introducing similar candidates (see also, e.g.,~\cite{ElkindFS12,liquidfeedback14,Brill18}).

    \item {\it reversal symmetry:} This criterion proposed by Saari~\cite{Saari94} means that if candidate $v$ is the unique winner, then $v$ must not be a winner in an election with all the preference orders inverted.

    \item {\it Smith criterion:} The Smith set is the smallest non-empty subset of candidates who (strictly) win pairwise comparison with every candidate not in the subset.
    The Smith criterion is satisfied when the set of winners is a subset of the Smith set~\cite{Smith73}.
\end{itemize}

The Schulze method is a voting rule used broadly in many organizations, e.g.,
the Wikimedia Foundation\footnote{{\it The Wikimedia Foundation elections to the Board of Trustees (2011)}, \url{https://meta.wikimedia.org/wiki/Wikimedia_Foundation_elections/Board_elections/2011/en}, [Online; accessed 23-June-2021].},
Debian Project\footnote{{\it Constitution for the {D}ebian {P}roject (v1.7)}, \url{https://www.debian.org/devel/constitution.en.html}, [Online; accessed 23-June-2021].} and, e.g,
in the software {\it LiquidFeedback}~\cite{liquidfeedback14} (see also a comprehensive list of users in~\cite[Section 1]{Schulze-arxiv18}).

In this paper we focus on computational aspects of the Schulze method.
In particular: 1) constructing a weighted majority graph; 2) checking whether a particular candidate is a winner; 3) finding a winner; 4) listing all winners.

Indicating a winner or a set of all winners is the main goal of a voting rule (tasks 3 and 4 above).
The definition of the Schulze method consists of two steps that can be studied separately---task 1 above is the first of the two steps; it is useful for other voting rules as well.
The decision version of the problem (task 2) is useful when considering the computational complexity of the problem~\cite{CsarLP18}, or when considering bribery, manipulation and control issues~\cite{GaspersKNW13,ParkesX12,MentonS13,HemaspaandraLM16} where one needs to verify whether a particular candidate would become a winner or non-winner.

Given a weighted majority graph, the most general task is to create a (weak) order over the alternatives (as a preference aggregation method).
The best known algorithm for this task, and also for tasks 3 and 4 above, has a running time of $\Oh(m^{2.69})$ in terms of the number of candidates $m$ (see, e.g.,~\cite{wikiwidest}).
This algorithm relies on fast matrix multiplication techniques that can be impractical.
The fastest algorithm without fast matrix multiplication (i.e. combinatorial) has cubic running time~\cite{Schulze03}.
While the running time is polynomial, for large-scale data this running time can be too slow.
To overcome this issue, parallel algorithms for the Schulze method were proposed, and it turns out that some tasks of the method have efficient parallel solutions (e.g. checking if a candidate is a winner subject to the Schulze method is in NL~\cite{CsarLP18}).

It is an interesting open question whether one can beat the known running time in the sequential setting, in particular without using fast matrix multiplication techniques.
It is interesting how fast one can find all winners when the weighted majority graph is given, and whether the majority graph itself can be computed efficiently.

\subsection{Our Contribution}

Let $m$ be the number of candidates, $n$ the number of voters.
Schulze~\cite{Schulze11} showed that the weighted majority graph can be constructed in $\Oh(nm^2)$ time and then the set of all winners can be found in $\Oh(m^3)$ time.

\paragraph{Observations.}
We first make several warm-up observations, connecting the Schulze method to problems in graph algorithms.

The first observation (Proposition~\ref{prop:computing_weights}) is that the problem of constructing the weighted majority graph can be reduced to the so called {\it Dominance Product} problem studied by Matou\v{s}ek~\cite{MatIPL}.
This gives a new running time for the weighted majority graph construction problem of $\tO(n m^{(1+\omega)/2}+n^{(2\omega-4)/(\omega-1)} m^2)\leq \Oh(n m^{1.69}+n^{0.55}m^2)$, where $\omega<2.373$ is the exponent of square matrix multiplication~\cite{AlmanW20,Williams12,Gall14a} and the $\tO$ notation hides subpolynomial factors.
This running time always beats the previously published running time of $\Oh(nm^2)$ when the number of voters and candidates is super-constant.

The second observation (Proposition~\ref{prop:computing_winners_slow}) is folklore (see, e.g., the Wikipedia article~\cite{wikiwidest}): computing the set of all winners (when the weighted majority graph is given) can easily be reduced to the so called All-Pairs Bottleneck Paths (APBP) problem:
given a graph $G$, for every pair of vertices $u,v$, compute the maximum over all $u$-$v$ paths of the minimum edge weight on the path (the so-called {\em bottleneck}).
Using the fastest known algorithm for APBP~\cite{DuanP09}, one can compute all winners in $\tO(m^{(3+\omega)/2})\leq \Oh(m^{2.69})$ time, easily beating the cubic running time.
Actually, after solving APBP we obtain indirect comparisons between all pairs of candidates which can be used to construct a weak order over all candidates.
This is useful in providing a fixed number of winners or as a preference aggregation method.

\paragraph{New improved algorithms.}
Although the above observations are of mostly theoretical interest, they do suggest that further algorithmic improvements can be possible.
We turn our attention to obtaining even faster algorithms.

Our first main contribution is an almost {\em quadratic} time algorithm for finding all winners.
This running time is substantially faster than all previously known algorithms.
Furthermore, as the running time is almost linear in the size of the weighted majority graph, it is {\em essentially optimal}.
The algorithm is combinatorial (it does not use heavy techniques that have large overheads like matrix multiplication), and is potentially of practical interest.

\begin{theorem}
Given a weighted majority graph on $m$ candidates, one can compute the set of all winners of the Schulze method in expected $\Oh(m^2 \log^4(m))$ time.
\end{theorem}

The theorem above appears as Theorem~\ref{thm:schulze-allwinner-n2log4-n} in the body of the paper.
As a warm-up to Theorem~\ref{thm:schulze-allwinner-n2log4-n} we first present, in Theorem~\ref{thm:schulze-winner-n2log4-n}, a combinatorial algorithm for finding a single winner and then generalize the approach to finding all winners.
(Note that the problem of verifying that a particular candidate is a winner can easily be solved in $\Oh(m^2)$ time using known algorithms, when given a weighted majority graph (see Proposition~\ref{prop:verying_winner}).
However, computing the winners is a more difficult problem.)

\paragraph{Fine-grained lower bounds.}
While we were able to achieve an essentially optimal $\tO(m^2)$ time algorithm for finding all winners when given a weighted majority graph, computing the weighted majority graph itself seems more expensive.
The fastest algorithm comes from our simple reduction to Dominance Product, resulting in a running time of $\tO(n m^{(1+\omega)/2}+n^{(2\omega-4)/(\omega-1)} m^2)$, or very slightly faster if $\omega>2$ and one uses fast rectangular matrix multiplication.\footnote{If $\omega=2$, using rectangular matrix multiplication gives no improvement.}

Typically, the number of voters $n$ is no smaller than the number of candidates $m$.
In this case, the running time for computing the majority graph using our reduction to Dominance Product is at least $\Omega(m^{2.5})$, regardless of the value of $\omega$.
Thus, if the input to the Schulze winner problem is not the weighted majority graph, but the preference profile, then the $\Omega(m^{2.5}) \gg m^2$ running time for computing the weighted majority graph is a significant bottleneck.

This leads to the following questions:
\begin{enumerate}
\item Can one compute the weighted majority graph in $\tO(m^2)$ time, or at least in $\Oh(m^{2.5-\eps})$ time for some $\eps>0$?
\item If not, can a winner be found in $\Oh(m^{2.5-\eps})$ time for some $\eps>0$, given the preference profile, potentially without computing the weighted majority graph explicitly?
\end{enumerate}

We address the above questions through the lens of {\em fine-grained complexity} (see the surveys~\cite{vsurvey,aviadvsurvey}).
The main goal of fine-grained complexity is to relate seemingly different problems via fine-grained reductions, hopefully proving equivalences.
In recent years there has been an explosion of results in fine-grained complexity, showing that a large variety of problems are fine-grained reducible to each other.
Sometimes fine-grained results can be viewed as hardness results, in the sense that a running time improvement for a problem can be considered unlikely.
Fine-grained equivalences are especially valuable, as they often establish strong connections between seemingly unrelated problems.

For running time functions $a(n)$ and $b(n)$ for inputs of size (or measure\footnote{In fact, in fine-grained complexity the running time is sometimes measured in terms of a different measure than the input size.
For instance, the running time of graph problems is often measured in terms of the number of vertices instead of the true size of the graph which is in terms of the number of vertices and edges.})
$n$, we say that there is an $(a,b)$-fine-grained reduction from a problem $A$ and to a problem $B$ if for every $\eps>0$ there is a $\delta>0$ and an algorithm that runs in $\Oh(a(n)^{1-\delta})$ time and solves a size $n$ instance of problem $A$ using oracle calls to problem $B$ of sizes $n_1,\ldots,n_k$ so that $\sum_{j=1}^k b(n_j)^{1-\eps} \leq a(n)^{1-\delta}$.
In other words, any improvement {\em in the exponent} over the $\Oh(b(n))$ runtime for $B$ on inputs of size $n$ can be translated to an improvement in the exponent over the $\Oh(a(n))$ runtime for $A$.

We relate questions (1) and (2) above to the complexity of the aforementioned Dominance Product problem and its relative {\it Dominating Pairs}.
These problems have been studied within both algorithms (e.g.~\cite{MatIPL,Yuster09}) and fine-grained complexity (e.g.~\cite{Labib19,Lincoln20}).

In the Dominance Product problem we are given two $r\times r$ matrices $A$ and $B$ and we want to compute for every pair $i,j\in [r]$ the number of $k\in [r]$ for which $A[i,k]\leq B[k,j]$.
In the Dominating Pairs problem we are also given two $r\times r$ matrices $A$ and $B$, but we now want to decide whether there exists a pair $i,j\in [r]$ such that for all $k\in [r]$, $A[i,k]\leq B[k,j]$.
In other words, we want to decide whether the Dominance Product of $A$ and $B$ contains an entry of value $r$.

Both problems can be solved using the aforementioned $\tO(r^{(3+\omega)/2})$ time algorithm by Matou\v{s}ek~\cite{MatIPL}.
If $\omega>2$, a slightly faster algorithm is possible using fast rectangular matrix multiplication~\cite{legallurrutia}, as shown by Yuster~\cite{Yuster09}.
These running times are minimized at $\tO(r^{2.5})$ (in the case that $\omega=2$).

We prove an equivalence between Dominance Product and the problem of computing the weighted majority graph, given a preference profile for the special case of $n$ voters and $m=\Oh(n)$ candidates.
Note that even $m=n$ is a natural case---when a group of voters chooses a winner among themselves.

\begin{theorem}
If there exists a $T(n)$ time algorithm for computing the weighted majority graph when there are $n$ voters and $\Oh(n)$ candidates, then there is an $\Oh(T(r) + r^2 \log r)$ time algorithm for computing the Dominance Product of two $r\times r$ matrices.
If there is a $T(r)$ time algorithm for computing the Dominance Product of two $r\times r$ matrices, then one can compute the weighted majority graph for $n$ voters and $\Oh(n)$ candidates in $\Oh(T(n))$ time.
\end{theorem}

First notice that since the matrices in the Dominance Product themselves have sizes $\Oh(r^2)$, the $\Oh(r^2\log r)$ running time in the first part of the theorem is necessary, up to the $\log r$ factor.
The first part of the theorem is proven in Theorem~\ref{thm:lb_graph} and the second part follows from our simple observation described in Proposition~\ref{prop:computing_weights}.

As the $r^{2.5-o(1)}$ time bottleneck for solving Dominance Product and even Dominating Pairs has remained unchallenged for 30 years, it is believed that the two problems require $r^{2.5-o(1)}$ time (see, e.g.,~\cite{Lincoln20}).
Due to Theorem~\ref{thm:lb_graph} (in which $m=2n$), we get that it would likely be very challenging to obtain an $\Oh(m^{2.5-\eps})$ time algorithm for $\eps>0$ for computing the weighted majority graph, even if $\omega=2$.

We have shown that constructing the weighted majority graph is likely an expensive task.
It could still be, however, that one can find a winner without explicitly constructing the weighted majority graph, as per our question (2) above.
Our final result relates question (2) to the Dominating Pairs problem:

\begin{theorem}
Suppose that there is an $\Oh(T(n))$ time algorithm that, given $n$ voters with preferences over $\Oh(n)$ candidates, can test whether a given candidate is a winner (or can find an arbitrary winner) of the Schulze method.
Then there is an $\Oh(T(r)+r^2\log r)$ time algorithm for the Dominating Pairs problem for two $r\times r$ matrices.
\end{theorem}

The theorem above appears as Theorem~\ref{thm:lb_testing_winner} in the body of the paper.
Because in the proof we have $m=\Theta(n)$, our result shows that even testing whether a candidate is a winner in $\Oh(m^{2.5-\eps})$ time for $\eps>0$ would be a challenge under the plausible hypothesis that Dominating Pairs has no improved algorithms even in the case that $\omega=2$.
Under this hypothesis, we might as well compute the weighted majority graph and find all winners using Theorem~\ref{thm:schulze-allwinner-n2log4-n}, as it would take roughly the same time.

Our paper is the {\em first} to bring fine-grained complexity into the field of computational social choice (see~\cite{handbookcomsoc2016,Endriss17trends,AzizBES19} for a description of computational social choice topics).
Related areas such as Fixed Parameter Tractability (FPT)\footnote{
Where some standard assumptions are: $\textsf{FPT} \neq \textsf{W}[1]$ and the Exponential Time Hypothesis.
For more about parameterized complexity see, e.g.,~\cite{DowneyF99,CyganFKLMPPS15}.}
and hardness of approximation\footnote{
Here some standard assumptions are $\textsf{P} \neq \textsf{NP}$ , the Unique Games Conjecture~\cite{Khot02a}, and a recent Gap-Exponential Time Hypothesis~\cite{FeldmannSLM20}.}
have already gained significant traction in computational social choice.
See, e.g., the following surveys on FPT results and further challenges in computational social choice~\cite{bredereck2014parameterized,FaliszewskiN16fptcomsoc,DornS17trendsfpt},
and, e.g., the following papers on hardness of approximation in bribery problems~\cite{FaliszewskiMS21},
multiwinner elections~\cite{SkowronFL16,BarmanFGG20,DudyczMMS20,BarmanFF21} and a survey on fair-division with indivisible goods~\cite{Markakis17}.

While these related areas do discuss running time, they do not focus on the exact running time complexity, and in particular, say nothing about lower bounds for problems that already have fixed polynomial time algorithms.
By bringing fine-grained complexity into the picture we can identify which voting protocols have potentially practical algorithms, and those whose complexity is, say at least cubic, and hence are likely impractical for large numbers of candidates and/or voters.

Our results provide reductions and equivalences that are even tighter than typical fine-grained reductions, in that the running times are preserved up to small additive terms and constant multiplicative factors.

\subsection{Related Work}

Below we present a few selected computational topics considered in the context of the Schulze method.
For a detailed survey of related work and historical notes, we refer the reader to a technical report by Schulze~\cite{Schulze-arxiv18}.
In particular, we will not discuss the axiomatic properties of the Schulze method; a comprehensive study of this appears in~\cite[Section 4]{Schulze11} and its extended version~\cite[Section 4]{Schulze-arxiv18}.

\paragraph{Parallel computing.}
Csar et al.~\cite{CsarLP18} considered large-scale data sets and claimed that a straight-forward implementation of the Schulze method does not scale well for large elections.
Because of this, they developed an algorithm in the massively parallel computation model.
They showed that the Schulze winner determination problem is NL-complete.
The containment in NL implies that the method is well-suited for parallel computation.
Csar et al.~\cite{CsarLP18} also conducted experiments to show that their algorithm scales well with additional computational resources.

\paragraph{Strategic behavior.}
There are 3 basic types of strategic behaviors: manipulation, control and bribery.

{\it Manipulation} is defined as the casting of non-truthful preference ballots by a voter or a coalition of voters in order to change the outcome of an election (to make a preferred candidate a winner is called ``constructive manipulation'', and to make a particular candidate a loser is called ``destructive manipulation'').

{\it Control} has eight basic variants, each defined by picking a choice from each of the following three pairs of groups:
$\{$constructive, destructive$\}$, $\{$adding, deleting$\}$, $\{$votes, alternatives$\}$.; e.g. ``constructive control by deleting votes''.

{\it Bribery} allows one to completely change votes, however, there is a cost of changing each vote, so that the number of affected votes is to be minimized.
There is a constructive and a destructive version of bribery.

Parkes and Xia~\cite{ParkesX12} showed that the Schulze method is vulnerable (polynomial-time solvable) to constructive manipulation by a single voter and destructive manipulation by a coalition.
From the other side, the Schulze method is resistant (NP-hard to compute a solution) to, e.g.,
constructive control by adding alternatives (i.e., it is NP-hard to find a subset of additional alternatives such that a preferred candidate becomes a winner);
control by adding/deleting votes;
and both cases of bribery.
For more specific results on computational issues of strategic behaviors under the Schulze method see~\cite{MentonS13,GaspersKNW13}.
Furthermore, FPT algorithms for NP-hard types of strategic behaviors were studied by Hemaspaandra et al.~\cite{HemaspaandraLM16}.

\paragraph{Margin of victory.}
The Schulze method was studied in terms of the {\it margin of victory}, which is the minimum number of modified votes resulting in a change to the set of winners.
This is a similar concept to bribery, but here we only care about the stability of a solution, not about making a candidate a winner/loser.
This concept is used to measure the robustness of voting systems due to errors or frauds.
The higher the margin of victory is, the more robust the election is.
Reisch et al.~\cite{ReischRS14} showed, e.g., that computing the margin of victory for the Schulze method is NP-hard
(using the NP-hardness proof for destructive bribery due to Parkes and Xia~\cite{ParkesX12}), so that evaluating the robustness of the election might be difficult.

\section{Preliminaries}\label{sec:preliminaries}

To define the method formally we need to introduce some notations first.

{\it A weak order} (or {\it total preorder}) $\succeq$ is a binary relation on a set $S$ that is transitive and connex (complete), i.e.,
(1) if $x \succeq y$ and $y \succeq z$ then $x \succeq z$;
(2) for all $x,y \in S$ we have $x \succeq y$ or $y \succeq x$.
We define a strict part $\succ$ of a weak order $\succeq$ by: $x \succ y$ iff $x \succeq y$ and $y \nsucceq x$.
Then $\succ$ is a {\it strict weak order}, and the Schulze method was originally defined using it~\cite{Schulze11}.

We denote the set of voters by $N$, and the set of alternatives (candidates) by $A$, where $|N| = n$ and $|A|=m$.
{\it A preference profile} $P$ (or a multiset of votes) is a list $(\succeq_a)_{a \in N}$ of weak orders over a set of alternatives $A$.
$u \succ_a v$ ($u \succeq_a v$) means that a voter $a \in N$ strictly (weakly respectively) prefers alternative $u \in A$ to alternative $v \in A\setminus\{u\}$.

We define $M(u,v)$, where $u,v \in A$, as the number of voters that strictly prefer $u$ over $v$, i.e., $M(u,v) = |\{a \in N: u \succ_a v\}|$.

For a given preference profile $P$, the {\em weighted majority graph} $G$ is defined as follows:
the vertices of $G$ is the set of alternatives $A$, and for every $u,v \in A$ we have directed edges $(u,v)$ and $(v,u)$ with associated weights
$w(u,v) = M(u,v) - M(v,u)$ and
$w(v,u) = M(v,u) - M(u,v)$.
We also call $w(u,v)$ {\it the strength of the link $(u,v)$}.
We note that Schulze~\cite{Schulze11} defined the strength of the link in a more general way, but he proposed to use $w(u,v)$ defined above, as it is the most intuitive notion of strength.
Indeed, this is the most popular strength of the link definition used in the literature~\cite{ParkesX12,GaspersKNW13,MentonS13,ReischRS14,HemaspaandraLM16,CsarLP18}.

We define the {\it strength of indirect comparison} of alternative $u \in A$ versus alternative $v \in A \setminus \{u\}$,
denoted by $B_G(u,v) \in \{0,1,\dots,n\}$,
as the weight of {\it the maximum bottleneck path} (also called widest path) from $u$ to $v$.
$B_G(u,v)$ is the maximum {\it width} or {\em bottleneck} of any path from $u$ to $v$, where the width/bottleneck of a path is equal to its minimum edge weight.
Formally $B_G(u,v) = \max_{(x_1=u,x_2,\dots,x_{k-1},x_k=v): x_j \in A} \min_{i \in \{2,3,\dots,k\}} \{ w(x_{i-1},x_i) \}$.
(The maximum is well-defined because widest paths are simple without loss of generality, similar to shortest paths.)
If $G$ is clear from context, we sometimes use $B(u, v)$ as well.

A set of winners $\calW$ in the Schulze method consists of all $u \in A$ such that for every $v \in A$ we have $B_G(u,v) \geq B_G(v,u)$.
It is known that there always exists a winner, i.e., $\calW \neq \emptyset$~\cite{Schulze11}.
We call the elements of $\calW$ {\it the Schulze winners}.
By $\calW(G)$ we denote a set of winners for a given weighted majority graph $G$.

\schulzeaw is the problem of finding all winners (i.e., the set $\calW$),
and \schulzew is that of finding a winner (i.e., an element from $\calW$).
The decision version of the problem, \schulzewd~\cite{CsarLP18}, asks whether a given candidate is a winner.

The Schulze method was designed as a single winner election rule,
but using it we can construct a weak order over the set of all alternatives.
Hence the Schulze method may be seen as a preference aggregation method.
For this we define the relation $R$ such that $(u,v) \in R$ if and only if $B_G(u,v) > B_G(v,u)$.
Note that the Schulze winners are top-ranked alternatives in the order derived from $R$.
Also note that a proof of $\calW \neq \emptyset$ follows from transitivity of $R$~\cite[Section 4.1]{Schulze11}.

We use $\omega$ to denote the smallest real number such that one can multiply two $r \times r$ matrices in $\Oh(r^{\omega+\epsilon})$ time for every $\epsilon>0$.
Currently we know that $2 \leq \omega < 2.373$~\cite{Williams12,Gall14a,AlmanW20}.
We also use $\mathcal{M}(a, b, c)$ to denote the fastest running time for multiplying an $a \times b$ matrix and a $b \times c$ matrix.

For a graph $G$ and a subset of vertices $U \subseteq V(G)$, we use $G[U]$ to denote the subgraph induced by the vertex set $U$.
We use \emph{strongly-connected-component (SCC) of a vertex $v$} to denote the set of vertices that can both reach and be reached from $v$.

\section{Warm-up}
\label{sec:warmup}

\begin{proposition}
\label{prop:computing_weights}
For a preference profile with $m$ candidates and $n$ voters, we can compute the weighted majority graph in $\min_s \tO(\mathcal{M}(m, sn, m) + nm^2/s)$ time.
\end{proposition}
\begin{proof}
We will compute the weighted majority graph using the algorithm for Dominance Product~\cite{Yuster09}, which runs in time $\min_s \tO(\mathcal{M}(m, sn, m) + nm^2/s)$.

Given $m$ candidates and $n$ voters, we will create the matrices $A$ and $B$ as follows.
Since the preference list of each voter $a$ is a weak order, we can associate an integer value $f_{a,u}$ with each voter $a$ and each candidate $u$ such that $u \succ_a v$ if and only if $f_{a, u} < f_{a, v}$ for any two candidates $u, v$ (these values can be the ranks in the sorted order of $a$'s preference list).
We then set $A_{u, a} = f_{a, u}$ for any pair consisting of a voter $a$ and a candidate $u$; we set $B_{a, v} = f_{a, v} - \frac{1}{2}$ for any pair consisting of a voter $a$ and a candidate $v$.
Then the Dominance Product $C$ between $A$ and $B$ is such that
\[C_{u, v} = \left|\left\{ a \in N: f_{a,u} \le f_{a, v} - \frac{1}{2} \right\}\right|, \]
which equals exactly $M(u, v)$ by the definition of the integer values $f$.
Using $M$, we can compute the weighted majority graph in $\Oh(m^2)$ time.

Therefore, the bottleneck of the algorithm is the Dominance Product problem, which has running time $\min_s \tO(\mathcal{M}(m, sn, m) + nm^2/s)$.
\end{proof}

Suppose that $n \ge m^{(\omega -1)/2}$.
Then we set $s$ to a value such that $sn \ge m$.
In this case we can compute the product between an $m \times (sn)$ matrix and an $(sn) \times m$ matrix by first splitting the second dimension of the first matrix to roughly $\frac{sn}{m}$ pieces, and then using fast square matrix multiplication to compute the product between each pair of corresponding pieces.
Thus, we can upper bound $\mathcal{M}(m, sn, m)$ by $\tO((sn/m) \cdot m^\omega)$.
By setting $s = m^{(3-\omega)/2}$, the running time of Proposition~\ref{prop:computing_weights} becomes $\tO(nm^{(1+\omega) / 2}) \le \Oh(nm^{1.69})$.

If $n \le m^{(\omega -1)/2} $, then we set $s$ so that $sn \le m$.
We can similarly bound $\mathcal{M}(m, sn, m)$ by $\tO((m/sn)^2 \cdot (sn)^\omega)$.
By setting $s = n^{(3-\omega)/(\omega-1)}$, the running time of Proposition~\ref{prop:computing_weights} becomes $\tO( n^{(2\omega-4)/(\omega-1)} \cdot m^2) \le \Oh(n^{0.55} m^2)$.

Overall, the running time of Proposition~\ref{prop:computing_weights} is always upper bounded by $\Oh(nm^{1.69} + n^{0.55} m^2)$.
If $\omega>2$ the running time can be improved slightly by using the best known bounds on rectangular matrix multiplication~\cite{legallurrutia} to compute $\mathcal{M}(m, sn, m)$ in both cases.
We will not go into detail here because the setting of $s$ here would depend on how $n$ and $m$ are related, and the current best bounds on rectangular matrix multiplication are obtained by using numerical solvers for each setting of the matrix dimensions.
As mentioned in the Our Contribution subsection, if $n\geq m$, the running time is always greater than $m^{2.5-o(1)}$, regardless of the value of $\omega$ and the use of rectangular matrix multiplication.

The next two propositions are folklore~\cite{wikiwidest}. We include their proofs here for completeness.

\begin{proposition}
\label{prop:computing_winners_slow}
Given a weighted majority graph on $m$ candidates, \schulzeaw can be solved in $\tO(m^{(3+\omega)/2})$ time.
\end{proposition}
\begin{proof}
Given a weighted majority graph, all values $B_G(u, v)$ can be obtained via a single All-Pairs Bottleneck Paths (APBP) computation.
APBP has been well-studied by the graph algorithms community~\cite{ShapiraYZ11,VassilevskaWY09,DuanP09} and its current best running time is $\tO(m^{\frac{3+\omega}{2}})$ for an $m$-vertex graph~\cite{DuanP09}.
After we compute $B_G(u, v)$, it only takes $\Oh(m^2)$ additional time to determine all Schulze winners.
\end{proof}

\begin{proposition}
\label{prop:verying_winner}
Given a weighted majority graph on $m$ candidates and a particular candidate $v$, we can solve \schulzewd in $\Oh(m^2)$ time.
\end{proposition}
\begin{proof}
Since we only need to verify if some candidate $v$ is a winner, it suffices to compute $B_G(v, u)$ and $B_G(u, v)$ for all $u \in V$.
Computing $B_G(v, u)$ is exactly the problem of computing Single-Source Bottleneck Paths (SSBP).
In a dense graph with $\Theta(m^2)$ edges, the best algorithm for SSBP runs in $\Oh(m^2)$ time by using Dijkstra's algorithm augmented with Fibonacci heap.
To compute $B_G(u, v)$, we can reverse the directions of all edges in the graph and compute another SSBP.
Using $B_G(v, u)$ and $B_G(u, v)$ for all $u \in V$, it only takes $\Oh(m)$ time to determine if $v$ is a winner.
\end{proof}

\section{Finding a Winner}

In this section we show that given a weighted majority graph, we can find a Schulze winner in almost quadratic time.
\begin{theorem}\label{thm:schulze-winner-n2log4-n}
 Given a weighted majority graph on $m$ candidates, \schulzew can be solved in expected $\Oh(m^2 \log^4(m))$ time.
\end{theorem}

We will use the following decremental SCC algorithm of Bernstein et al.~\cite{bernstein2019decremental}.

\begin{theorem}\label{thm:SCC}
(Bernstein et al.~\cite{bernstein2019decremental})
Given a graph $G = (V, E)$,
we can maintain a data structure that supports the following operations:
\begin{itemize}
\item \texttt{DELETE-EDGE}$(u, v)$: Deletes the edge $(u, v)$ from the graph.
\item \texttt{SAME-SCC}$(u, v)$: Returns whether $u$ and $v$ are in the same SCC.
\end{itemize}
The data structure runs in total expected $\Oh(|E| \log^4 |V|)$ time for all deletions and worst-case $\Oh(1)$ time for each query.
The bound holds against an oblivious adaptive adversary.
\end{theorem}
Since the answer for each \texttt{SAME-SCC} query is unique, an oblivious adaptive adversary is equivalent to a non-adaptive adversary in this setting.

Our algorithm is simple and can be described in ten lines as shown in Algorithm~\ref{algo:one_winner}.

\begin{algorithm}
  \caption{\schulzew$(G=(V,E))$}\label{algo:one_winner}
  \Indentp{-1em}
    \KwData{$G=(V,E)$}
    \KwResult{One winner of graph $G$.}
  \Indentp{1em}
  Sort $E$ by weights in increasing order\;
  $x \gets $ an arbitrary vertex in $V$\;\label{alg1:2}
  \For{$(u, v) \in E$}
  {
    $\textit{flag} \gets \texttt{SAME-SCC}(u, x) \wedge \texttt{SAME-SCC}(v, x)$\;
    \texttt{DELETE-EDGE}$(u, v)$\;\label{alg1:5}
    \If{\texttt{SAME-SCC}$(u, v) = $ False \textbf{ and } \textit{flag}\label{alg1:6}}
    {
      $x \gets v$\;\label{alg1:7}
    }\label{alg1:8}
  }
  \KwRet{$x$}\;\label{alg1:10}
\end{algorithm}

\begin{lemma}
\label{lem:one_winner_induction}
Any time Algorithm~\ref{algo:one_winner} finishes Line~\ref{alg1:2} or Line~\ref{alg1:8}, any winner of the subgraph induced by the SCC of $x$ is a winner of the original graph $G$.
\end{lemma}

\begin{proof}

Even though Algorithm~\ref{algo:one_winner} seems to operate directly on graph $G$, in the proof we assume it operates on a copy of the graph and thus the original graph still has all its edges.
Thus, we always use $G$ to denote the original graph with no edges removed in this proof.

When the algorithm finishes Line~\ref{alg1:2}, all edges in the graph still exist.
Since the graph is complete, the SCC of $x$ is exactly $V$.
Therefore, any winner of the subgraph induced by the SCC of $x$ is a winner of the whole graph.

In the remainder of this proof, we will use $U$ to denote the set of vertices in the same SCC as $x$ in the graph where we remove all the edges in sorted list $E$ before $(u, v)$ (i.e., before we execute Line~\ref{alg1:5}).
We also use $U'$ to denote the set of vertices in the same SCC as $x$ in the graph where we remove all the edges up to $(u, v)$ and we considered updating $x$ in Line~\ref{alg1:7} under certain conditions (i.e., after we execute Line~\ref{alg1:8}).

We prove the second part of the lemma by induction.
Suppose $\mathcal{W}(G[U]) \subseteq \mathcal{W}(G)$, and we need to prove $\mathcal{W}(G[U']) \subseteq \mathcal{W}(G)$.
To achieve this, it suffices to show $\mathcal{W}(G[U']) \subseteq \mathcal{W}(G[U])$.

First of all, if \textit{flag} is set to \textit{False}, then either $u$ or $v$ is not in $U$, so the edge $(u, v)$ does not lie entirely in the set $U$.
Thus, deleting the edge $(u, v)$ does not change the SCC of $x$.
Also, since \textit{flag} is set to \textit{False}, the algorithm will not execute Line~\ref{alg1:7}, so the value of $x$ also stays the same.
Thus, since both $x$ and the SCC of $x$ stays unchanged, $U=U'$ and clearly $\mathcal{W}(G[U']) \subseteq \mathcal{W}(G[U])$.

Secondly, suppose the \texttt{SAME-SCC}$(u, v)$ check at Line~\ref{alg1:6} is \textit{True}.
This means $u$ and $v$ are still in the same SCC after deleting $(u, v)$, so in particular, $u$ can still reach $v$.
Thus, even though we just deleted the edge $(u, v)$, the connectivity of the graph is unaffected.
Thus, the SCC of $x$ is unchanged.
Also, the \texttt{SAME-SCC}$(u, v)$ check is \textit{True}, so we will not execute Line~\ref{alg1:7} in this iteration.
Therefore, similar to the previous case, $\mathcal{W}(G[U']) \subseteq \mathcal{W}(G[U])$.

The only case remaining is when \textit{flag} is set to \textit{True} and the \texttt{SAME-SCC}$(u, v)$ check at Line~\ref{alg1:6} is \textit{False}.
Let $y$ be any winner of the graph $G[U']$.
We will show that $y$ is a winner of $G[U]$ as well, and by the induction hypothesis $\mathcal{W}(G[U]) \subseteq \mathcal{W}(G)$, $y$ will be a winner of $G$.

We have to show that $B_{G[U]}(y, z) \ge B_{G[U]}(z, y)$ for every $z \in U$.
There are two cases depending on where $z$ is included.
\begin{enumerate}
    \item $z \in U'$.
          Since $y$ is a winner of $G[U']$, we have $B_{G[U']}(y, z) \ge B_{G[U']}(z, y)$.
          Now consider any path from $y$ to $z$ in $G[U]$.
          If this path ever touches any vertex $z'$ outside of $U'$, then it must use an edge that is already deleted, since otherwise, $y \in U'$ can reach $z'$ and  $z'$ can reach $z \in U'$ so $z'$ must also be in the SCC $U$, which leads to a contradiction.
          Thus, this path must use an edge of weight at most $w(u, v)$, so this path has a bottleneck at most $w(u, v)$.\footnote{Recall that $G[U]$ is an induced subgraph of $G$ in which all the edges are present---also these removed before removing $(u, v)$---hence a bottleneck of the considered path can be strictly smaller than $w(u, v)$.}
          However, since $U'$ is strongly connected, and all edges that are still present have weights at least $w(u, v)$, so $B_{G[U']}(y, z) \ge w(u, v)$.
          Thus, $B_{G[U]}(y, z) = B_{G[U']}(y, z)$ since if a path leaves $U'$ then its bottleneck is at most $w(u, v) \le B_{G[U']}(y, z)$.
          Similarly, $B_{G[U]}(z, y) = B_{G[U']}(z, y)$.
          Therefore, $B_{G[U]}(y, z) \ge B_{G[U]}(z, y)$.
    \item $z \in U \setminus U'$.
          We first show that $y$ can reach $z$ using only edges that are not yet deleted right after we delete $(u, v)$.
          First, since $y \in U'$, $y$ can reach $v$.
          Since $U$ is an SCC before we delete $(u, v)$, $v$ can reach all vertices in $U$ before we delete $(u, v)$.
          However, any simple path from $v$ to some other vertex in $U$ does not use the edge $(u, v)$, so $v$ can still reach all other vertices in $U$ even after deleting $(u, v)$.
          Therefore, $y$ can reach $z$ through $v$, and thus $B_{G[U]}(y, z) \ge w(u, v)$.
          On the other hand, $z$ cannot reach $y$ using not yet deleted edges since otherwise $z$ will be in the same SCC as $y$.
          Therefore, $B_{G[U]}(z, y) \le w(u, v) \le B_{G[U]}(y, z)$.
\end{enumerate}
\end{proof}

\begin{lemma}
\label{lem:one_winner_correct}
 Algorithm~\ref{algo:one_winner} always returns a winner of $G$.
\end{lemma}
\begin{proof}
By Lemma~\ref{lem:one_winner_induction}, at Line~\ref{alg1:10} of Algorithm~\ref{algo:one_winner}, any winner of the graph induced by the SCC of $x$ is a winner of $G$.
Furthermore, at Line~\ref{alg1:10}, we already deleted all edges in the graph, so the SCC of $x$ just contains $x$ itself.
Thus, $x$ is a winner of the graph induced by the SCC of $x$ and, by Lemma~\ref{lem:one_winner_induction}, $x$ is a winner of $G$.
\end{proof}

\begin{lemma}
\label{lem:one_winner_time}
 Algorithm~\ref{algo:one_winner} runs in expected $\Oh(m^2 \log^4(m))$ time.
\end{lemma}
\begin{proof}
Sorting the edge list $E$ takes $\Oh(m^2\log m)$ time.
Calling \texttt{DELETE-EDGE}$(u, v)$ for all $(u, v) \in E$ takes expected $\Oh(m^2 \log^4 (m))$ time by Theorem~\ref{thm:SCC} since we fix the order to delete the edges in advance and thus our algorithm behaves like a non-adaptive adversary.
Also, each call of \texttt{SAME-SCC} takes $\Oh(1)$ time by Theorem~\ref{thm:SCC}.
All remaining components of Algorithm~\ref{algo:one_winner} takes $\Oh(m^2)$ time.
Therefore, Algorithm~\ref{algo:one_winner} runs in expected $\Oh(m^2 \log^4(m))$ time.
\end{proof}

\begin{proof}[Proof of Theorem~\ref{thm:schulze-winner-n2log4-n}]
The theorem follows immediately from Lemma~\ref{lem:one_winner_correct} and Lemma~\ref{lem:one_winner_time}.
\end{proof}

\section{Finding All Winners}
In order to compute all winners, we need to augment the decremental SCC algorithm with more information.
This could be done in a black-box way.

In this section, we define the in-degree of an SCC $U$ as $\left| \{ (v, u) \in E: v \not \in U, u \in U\} \right|$.

\begin{corollary}
\label{cor:dec_scc}
Given a graph $G = (V, E)$,
we can maintain a data structure that keeps the following:
\begin{itemize}
    \item A set $\mathcal{S}$ containing all IDs of SCCs of the graph.
    \item A map $\mathcal{D}$ from the IDs of SCCs of the graph to the in-degrees of the SCCs.
    \item A map $\mathcal{SCC}$ from vertices of the graph to the ID of the SCCs they are in.
\end{itemize}
The data structure also supports the following operations:
\begin{itemize}
\item \texttt{DELETE-EDGE}$(u, v)$: Deletes the edge $(u, v)$ from the graph.
Additionally, the data structure needs to return a list of new SCCs being created.
\item \texttt{SAME-SCC}$(u, v)$: Returns whether $u$ and $v$ are in the same SCC.
\end{itemize}
The data structure runs in total expected $\Oh(|E| \log^4 |V|)$ time for all deletions and worst-case $\Oh(1)$ query time.
The bound holds against a non-adaptive adversary.
\end{corollary}

\begin{proof}
The high-level strategy of the algorithm is to first use the data structure from~\cite{bernstein2019decremental}, and then use the ``removing small from large'' strategy used in, e.g.,~\cite{shiloach1981line,thorup1999decremental,bernstein2019decremental} to explicitly maintain all the SCCs.

Initially, it is easy to set up the set $\mathcal{S}$ and the maps $\mathcal{D}$ and $\mathcal{SCC}$.
We also initialize the data structure of Bernstein et al.~\cite{bernstein2019decremental}.
For each vertex $v$, we create a list of vertices $\mathcal{N}(v)$ that contain all the neighbors of $v$ considering edges in both directions.

For each \texttt{DELETE-EDGE}$(u, v)$ operation, if $u$ and $v$ are not in the same SCC, it suffices to update the map $\mathcal{D}$.
If $u$ and $v$ are in the same SCC both before and after deleting the edge, we do not need to update $\mathcal{S}, \mathcal{D}$ or $\mathcal{SCC}$.

The trickiest case is when $u$ and $v$ are in the same SCC $U$ before deleting $(u, v)$, but not in the same SCC after deleting $(u, v)$.
In this case, we use the ``removing small from large'' method to find all new SCCs.
We will explicitly enumerate all the vertices in every new SCC except one SCC with the largest total degree of the vertices in it.

We first delete edge $(u, v)$ in the data structure of Bernstein et al.~\cite{bernstein2019decremental}.
Then we create a list of vertices $L$, which initially only contains $u$ and $v$.
Each time, we repeatedly take a vertex out of $L$ until we find one vertex $x$ whose new SCC has not been found.
Then we repeatedly take another vertex out of $L$ until we find one vertex $y$ which is not in the same new SCC as $x$ and whose new SCC has not been found.
If we could not find such a vertex $y$ before $L$ becomes empty, then we end the process.
We interleave two breadth-first-searches (BFSes) from vertex $x$ and vertex $y$ by using neighbors stored in $\mathcal{N}$, but exploring only those vertices which are in the same new SCC as $x$ or $y$ respectively by calling \texttt{SAME-SCC}.
Note that these BFSes use all edges in the original graph and ignore the edge directions.
As soon as one of the BFSes finishes (the SCC with the smaller total degree will finish sooner), we stop the other BFS as well.
Without loss of generality, we assume the BFS from $x$ finishes sooner.
In this case we have a list of all vertices in the same SCC with $x$.
For every vertex $w$ in this list, we enumerate all vertices $z$ in $\mathcal{N}(w)$, and add $z$ to $L$ if $z$ belonged to the SCC $U$ (we could call the map $\mathcal{SCC}$ for checking this, since it is not updated after $(u, v)$ gets deleted).
Finally we put $y$ back to $L$ and repeat the above process.

We show that the process above will find all new SCCs except the one SCC $U'$ with the largest total degree.
To show this, it suffices to show that all vertices in the induced subgraph $G[U \setminus U']$ are either weakly connected to $v$ or weakly connected to $u$ via a path inside $G[U \setminus U']$.
Let $z$ be any vertex in $U \setminus U'$.
Since before we delete $(u, v)$, $U$ is an SCC, so there is a simple path $p_1$ inside $U$ from $v$ to $z$ that only uses edges still in the graph including $(u, v)$.
However, since $p_1$ starts from $v$, and it is simple, it will not use $(u, v)$, so $p_1$ still exists even after deleting $(u, v)$.
Similarly, there is a path from $z$ to $u$ that still exists after deleting $(u, v)$.
If one of $p_1$ or $p_2$ does not enter $U'$, then we are done.
Otherwise, both $p_1$ and $p_2$ touch $U'$, so $z$ can both reach and be reached from $U'$.
Hence, $z$ also belongs to the SCC $U'$, a contradiction.
Therefore, our algorithm will find all but one new SCCs.

Since we can list all vertices in all but one new SCCs, it is then easy to update $\mathcal{S}, \mathcal{D}$ and $\mathcal{SCC}$.
However, we should note that the SCC $U'$ with the largest total degrees will have the same ID as the old SCC $U$, since we cannot afford to update the map $\mathcal{SCC}$ for every vertex in $U'$.
It is also easy to return a list of new SCCs being created.

Now we analyze the running time of the BFSes, which is quite standard.
Each time we find an SCC with total degree $D$, we will pay $\Oh(D)$ in the BFS and $\Oh(D \log |V|)$ for updating $\mathcal{S}, \mathcal{D}$ and $\mathcal{SCC}$.
We also know that we are taking it from an old SCC of total degree at least $2D$.
Therefore, each edge that contributes to $D$ can be taken out at most $\Oh(\log |V|)$ times, since each time the SCC that contains one endpoint of the edge must halve in total degree.
Therefore, the total cost of BFSes is $\Oh(|E| \log |V|)$, and there is another $\Oh(\log |V|)$ factor for updating the necessary data structures.
Thus, it is clear that the running time is dominated by the $\Oh(|E| \log^4 |V|)$ factor from~\cite{bernstein2019decremental}.
\end{proof}

Using Corollary~\ref{cor:dec_scc}, the algorithm for finding all winners is described in Algorithm~\ref{algo:all_winner}.

\begin{algorithm}
  \caption{\schulzeaw($G=(V,E)$)}\label{algo:all_winner}
  \Indentp{-1em}
    \KwData{$G=(V,E)$}
    \KwResult{All winners in graph $G$.}
  \Indentp{1em}
  Let $W$ be a sorted list of all distinct weights of $E$\;
  Let $\mathcal{C}$ be a set containing a single SCC, the whole graph\;
  \For{$w \in W$\label{alg2:3}}
  {
    Let $L$ be an empty list\;
    \For{every edge $(u, v)$ of weight $w$}
    {
        $\textit{scc} \gets \mathcal{SCC}(u)$\;
        $\textit{flag} \gets \texttt{SAME-SCC}(u, v)$\;
        $l \gets \texttt{DELETE-EDGE}(u, v)$\;
        \If{\texttt{SAME-SCC}$(u, v) = $ False \textbf{ and } \textit{flag} \textbf{ and } $\textit{scc} \in \mathcal{C}$}{
            Remove $scc$ from $\mathcal{C}$\;
            Add every SCC in $l$ to $\mathcal{C}$\;
            Add every SCC in $l$ to $L$\;
        }
    }
    \For{$\textit{scc} \in L$}
    {
        \If{$\mathcal{D}(\textit{scc}) > 0$}
        {
            Remove \textit{scc} from $\mathcal{C}$\;
        }
    }
  }
  \KwRet{$\{v \in V: \mathcal{SCC}(v) \in \mathcal{C}\}$}\;\label{alg2:21}
\end{algorithm}

We also describe the algorithm in text for more intuition.
The algorithm maintains a set $\mathcal{C}$ that contains all candidate SCCs that could contain winners.
In increasing order of weight $w$, the algorithm deletes all edges of weight $w$ in a batch.
If a candidate SCC splits into multiple smaller SCCs, we remove this candidate SCC from $\mathcal{C}$ and add those small SCCs whose in-degrees are $0$s back to $\mathcal{C}$.
Eventually, all SCCs contain single vertices, and we return the single vertices in those candidate SCCs.
Now we prove the correctness of the algorithm via the following lemma.

\begin{lemma}
\label{lem:all_winner_induction}
Before and after each iteration in the ``for'' loop in Line~\ref{alg2:3} of Algorithm~\ref{algo:all_winner}, we have
$$
\mathcal{W}(G) = \bigcup_{\textit{scc} \in \mathcal{C}} \mathcal{W}(G[\textit{scc}]),
$$
where $G$ denotes the original graph with no edge removed.
\end{lemma}
\begin{proof}
Before running any iteration of the ``for'' loop, $\mathcal{C}$ only contains the whole graph, so the equality is clearly true.

Now we prove the equality by induction.
Suppose that the equality is true before we run the ``for'' loop for value $w$.
For each vertex set $U \in \mathcal{C}$ that was an SCC before deleting all edges of weight $w$, we will split it into multiple SCCs after deleting those edges, and put those $U_1', \ldots, U_t'$ whose in-degrees are $0$s back to $\mathcal{C}$.
Thus, it suffices to show that
$$
\mathcal{W}(G[U]) = \bigcup_{i=1}^t \mathcal{W}(G[U_i']).
$$

Let $x \in U$ be an arbitrary winner of $G[U]$.
First, suppose the in-degree of $\mathcal{SCC}(x)$ after deleting the weight $w$ edges are nonzero.
In this case, there exists another vertex $y \in U \setminus \mathcal{SCC}(x)$ that can reach $x$.
Therefore, $B_{G[U]}(y, x) > w$.
However, $x$ cannot reach $y$ since otherwise $y$ is in the same SCC as $x$, so $B_{G[U]}(x, y) \le w$, and thus $x$ cannot be a winner of $G[U]$, a contradiction.
Therefore, we can assume $x \in U_i'$ for some $1 \le i \le t$.
For any $y \in U_i'$, consider the widest path between $x$ and $y$.
Since $x$ and $y$ belong to the same SCC after removing weight $w$ edges, $B_{G[U]}(x, y) > w$ and $B_{G[U]}(y, x) > w$.
Also, if any path leaves $U_i'$ and comes back, the bottleneck of that path is upper bounded by $w$, so an optimal path will not leave $U_i'$.
Therefore, $B_{G[U_i']}(x, y) = B_{G[U]}(x, y)$ and $B_{G[U_i']}(y, x) = B_{G[U]}(y, x)$.
Since $x$ is a winner of $G[U]$, $B_{G[U]}(x, y) \ge B_{G[U]}(y, x)$.
Therefore, $B_{G[U_i']}(x, y) \ge B_{G[U_i']}(y, x)$.
We conclude that $x$ is a winner of $G[U_i']$.
Hence, $\mathcal{W}(G[U]) \subseteq \bigcup_{i=1}^t \mathcal{W}(G[U_i'])$.

Now we show the other direction.
For any $1 \le i \le t$, let $x$ be an arbitrary winner of $G[U_i']$.
For any $y \in U$, we want to show that $B_{G[U]}(x, y) \ge B_{G[U]}(y, x)$.
First, if $y \not \in U_i'$, then there is no path from $y$ to $x$ using only edges of weight greater than $w$, because the in-degree of $U_i'$ is $0$.
Therefore, $B_{G[U]}(y, x) \le w$.
On the other hand, since $x, y \in U$ and $U$ is strongly connected using edges of weight up to $w$, we have $B_{G[U]}(x, y) \ge w$.
Thus, $B_{G[U]}(x, y) \ge B_{G[U]}(y, x)$.
Secondly, if $y \in U_i'$, then the widest paths from $x$ to $y$ and from $y$ to $x$ are completely inside $U_i'$ by a previous argument, so $B_{G[U_i']}(x, y) = B_{G[U]}(x, y)$ and $B_{G[U_i']}(y, x) = B_{G[U]}(y, x)$.
Since $x$ is a winner of $G[U_i']$, $B_{G[U_i']}(x, y) \ge B_{G[U_i']}(y, x)$.
Therefore, $B_{G[U]}(x, y) \ge B_{G[U]}(y, x)$.
We conclude that $x$ is a winner of $G[U]$ and hence $\mathcal{W}(G[U]) \supseteq \bigcup_{i=1}^t \mathcal{W}(G[U_i'])$.

In conclusion, we showed that $\mathcal{W}(G[U]) = \bigcup_{i=1}^t \mathcal{W}(G[U_i'])$, which is sufficient for the induction to complete.
\end{proof}

\begin{theorem}\label{thm:schulze-allwinner-n2log4-n}
 Given a weighted majority graph on $m$ candidates, \schulzeaw can be solved in expected $\Oh(m^2 \log^4(m))$ time.
\end{theorem}
\begin{proof}
The algorithm is shown in Algorithm~\ref{algo:all_winner}.
By Lemma~\ref{lem:all_winner_induction}, at Line~\ref{alg2:21}, we have $\mathcal{W}(G) = \bigcup_{\textit{scc} \in \mathcal{C}} \mathcal{W}(G[\textit{scc}])$.
Furthermore, since all edges are deleted in the graph and thus all SCCs contain single vertices at Line~\ref{alg2:21}, $\mathcal{W}(G[\textit{scc}]) = V(\textit{scc})$.
Therefore, $\mathcal{W}(G) = \{v \in V: \mathcal{SCC}(v) \in \mathcal{C}\}$ and thus the returned result is correct.

Algorithm~\ref{algo:all_winner} clearly runs in expected $\Oh(m^2 \log^4(m))$ time.
\end{proof}

\section{Lower Bounds}

\begin{theorem}
\label{thm:lb_graph}
If there exists a $T(n)$ time algorithm for computing all the edge weights of the weighted majority graph when there are $n$ voters and $2n$ candidates, then there is an $\Oh(T(r) + r^2 \log r)$ time algorithm for computing the Dominance Product of two $r\times r$ matrices.
\end{theorem}
\begin{proof}

Given two $r \times r$ matrices $A$ and $B$, we will compute their Dominance Product $C$, where $C_{i,j} = \left| k \in [r]: A_{i, k} \le B_{k, j}\right|$, using the assumed $T(n)$ time algorithm for computing the weighted majority graph.

We first pre-process the two matrices so that all entries are distinct.
We could achieve this by first putting all the entries to a list, and then sorting the list.
If there is a tie between several elements, we always sort an element corresponding to an entry of $A$ earlier than an element corresponding to an entry of $B$.
Then we can replace all entries with their position in the sorted list.
The Dominance Product of $A$ and $B$ clearly does not change.
This pre-processing only takes $\Oh(r^2 \log r)$ time.

In our construction, there are $m=2r$ candidates labeled as $u_1, \ldots, u_r$ and $v_1, \ldots, v_r$.
We will create one voter for each $k \in [r]$, so there are a total of $n = r$ voters.
The $k$-th voter associates a number $A_{i, k}$ with candidate $u_i$ for every $i \in [r]$ and a number $B_{k, j}$ with candidate $v_j$ for every $j \in [r]$.
Then the $k$-th voter prefers a candidate $x$ over a candidate $y$ if and only if the associated number of candidate $x$ is smaller than that of candidate $y$.
Since all entries of these two matrices are distinct, the preference order of each voter is linear (i.e., for every voter $i$ and any distinct candidates $x$ and $y$ we have either $x \succ_i y$ or $x \prec_i y$).

We show that the value $M(u_i, v_j)$ corresponding to preference profile above equals $C_{i, j}$.
In fact, if the $k$-th voter prefers $u_i$ to $v_j$, then its associated number with $u_i$ is smaller than its associated number with $v_j$, i.e., $A_{i, k} < B_{k, j}$.
Therefore, $M(u_i, v_j)$ equals the number of $k$ such that $A_{i, k} < B_{k, j}$.
Since all entries of the two matrices are distinct, $A_{i, k} < B_{k, j}$ if and only if $A_{i, k} \le B_{k, j}$.
Thus, $M(u_i, v_j) = C_{i, j}$.

In the weighted majority graph built on linear orders only, the edge weight $w(u_i, v_j)$ equals to $M(u_i, v_j) - M(v_j, u_i) = 2M(u_i, v_j) - r = 2C_{i,j} - r$, so we could compute $C_{i,j}$ from $w(u_i, v_j)$ via $C_{i,j} = (w(u_i, v_j) + r) / 2$.
Therefore, we can call the $T(n) = T(r)$ time algorithm for computing the weighted majority graph and get the Dominance Product $C$ from the edge weights of the graph easily.
\end{proof}

Next, we will show the conditional hardness for \schulzewd and \schulzew by reducing from the Dominating Pairs problem.
Note that in the Dominating Pairs problem, we essentially want to test if the Dominance Product of two $r \times r$ matrices contains an entry of value $r$.

\begin{theorem}
\label{thm:lb_testing_winner}
Suppose that there is an $\Oh(T(n))$ time algorithm that, given $n$ voters with preferences over $\Theta(n)$ candidates, can solve \schulzewd (or \schulzew).
Then there is an $\Oh(T(r)+r^2\log r)$ time algorithm for the Dominating Pairs problem for two $r\times r$ matrices.
\end{theorem}
\begin{proof}

Given two $r \times r$ matrices $A$ and $B$, we will create a preference profile on $m=2r+2$ candidates $u_1, \ldots, u_r, v_1, \ldots, v_r, W, W'$ and $n=10r-2$ voters.
For technical reasons, we assume all entries in the Dominance Product $C$ between $A$ and $B$ are positive.
This can easily be addressed by padding an additional column with $0$s to matrix $A$ and a corresponding row with $1$s to matrix $B$.
We can perform additional padding to make $A$ and $B$ square again:
add a row to $A$ that contains entries strictly greater than any entry in $B$ (except the last one entry in this row which equals to $0$);
add a column to $B$ that contains entries strictly smaller than any entry in $A$ (except the last one entry in this column which equals to $1$).
With an $\Oh(r^2 \log r)$ time pre-processing, we can assume all the entries of the matrices $A$ and $B$ are distinct (similarly as in the proof of Theorem~\ref{thm:lb_graph}).

The first $r$ voters will look like the voters in the proof of Theorem~\ref{thm:lb_graph}.
Specifically, the $k$-th voter associates a number $A_{i, k}$ with candidate $u_i$ for every $i \in [r]$, and a number $B_{k, j}$ with candidate $v_j$ for every $j \in [r]$.
Then the $k$-th voter prefers a candidate $x$ over a candidate $y$ if and only if the associated number of candidate $x$ is smaller than that of candidate $y$.
Moreover, the first $r$ voters always prefer $W$ the least and prefer $W'$ the second least.

The next $r$ voters will be similar to the first $r$ voters.
For any $k \in [r]$, the preference list of the $(k+r)$-th voter is the same as the preference list of the $k$-th voter, except that $(k+r)$-th voter always prefers $W$ the most and prefer $W'$ the second most.

Let us consider how the first $2r$ voters will affect the values of $M$.
From the first $2r$ voters, we add $2C_{i,j}$ to $M(u_i, v_j)$ for all $i, j \in [r]$, add $2r - 2C_{i,j}$ to $M(v_j, u_i)$ for all $i, j \in [r]$
and add $r$ to all edges with one endpoint being $W$ or $W'$.
We call these edges \emph{important} as other edge weights will not be important in our analysis.

The preference lists between the $(2r+1)$-th voter and the $(3r-1)$-th voter are all the following:
$$u_1 \prec u_2 \prec \cdots \prec u_r \prec W \prec v_1 \prec \cdots \prec v_r \prec W'.$$

The preference lists between the $(3r)$-th voter and the $(4r-2)$-th voter are all the following:
$$W' \prec v_r \prec \cdots \prec v_1 \prec u_r \prec u_{r-1} \prec \cdots \prec u_1 \prec W.$$

Note that from these two types of voters, we add $2r-2$ to $M(W, u_i)$ for all $i \in [r]$, add $0$ to $M(u_i, W)$ for all $i \in [r]$, and add $r-1$ to all other important edges.
We can apply the same idea for other voters and manipulate the edge weights as follows.
\begin{itemize}
    \item We can add $2r$ voters whose preference lists will add $2r$ to $M(W, W')$, add $0$ to $M(W', W)$, and add $r$ to all other edges.
          For this case, our construction is the same as the  McGarvey's method~\cite{mcgarvey1953theorem}.
    \item We can add $2r$ voters whose preference lists will add $2r$ to $M(W', v_j)$ for every $j \in [r]$, add $0$ to $M(v_j, W')$ for every $j \in [r]$, and add $r$ to all other edges.
    \item We can add $2r$ voters whose preference lists will add $2r$ to $M(v_j, W)$ for every $j \in [r]$, add $0$ to $M(W, v_j)$ for every $j \in [r]$, and add $r$ to all other edges.
\end{itemize}

Overall, the values $M(u, v)$ are summarized in Table~\ref{tab:lb_testing_winner_M}, and the edge weights of the weighted majority graph are summarized in Table~\ref{tab:lb_testing_winner_edge}.
\begin{table}[ht]
    \centering
    \begin{tabular}{|c|c|c|c|c|}
    \hline
     \diagbox{$u$}{$v$} & $W$ & $W'$ & $u_i$ & $v_j$\\
    \hline
        $W$ & $\star$ & $6r-1$& $6r-2$ & $4r-1$\\
    \hline
        $W'$ & $4r-1$ & $\star$ & $5r-1$ & $6r-1$\\
    \hline
        $u_i$ & $4r$ & $5r-1$ & $\star$ & $2C_{i,j}+4r-1$\\
    \hline
        $v_j$& $6r-1$& $4r-1$ & $6r-1-2C_{i,j}$ & $\star$ \\
    \hline
    \end{tabular}
    \caption{The values $M(u, v)$.
    The entries marked as $\star$ are not important in our analysis.}
    \label{tab:lb_testing_winner_M}
\end{table}

\begin{table}[ht]
    \centering
    \begin{tabular}{|c|c|c|c|c|}
    \hline
     \diagbox{$u$}{$v$} &$W$ & $W'$ & $u_i$ & $v_j$\\
    \hline
        $W$ & $\star$ & $2r$& $2r-2$ & $-2r$\\
    \hline
        $W'$ & $-2r$ & $\star$ & $0$ & $2r$\\
    \hline
        $u_i$ & $-2r+2$ & $0$ & $\star$ & $4C_{i,j}-2r$\\
    \hline
        $v_j$& $2r$& $-2r$ & $2r-4C_{i,j}$ & $\star$ \\
    \hline
    \end{tabular}
    \caption{The weights $w(u, v)$ in the weighted majority graph.
    The entries marked as $\star$ are not important in our analysis.}
    \label{tab:lb_testing_winner_edge}
\end{table}

\newpage 
\begin{claim}
\label{cl:lb_testing_winner}
The Dominance Product $C$ between $A$ and $B$ contains an entry of value $r$ if and only if $W$ is not a Schulze winner in the preference profile described above.
\end{claim}
\begin{proof}
Suppose $C$ contains an entry $C_{i,j}$ where $C_{i,j}=r$ for some $i, j \in [r]$.
In this case, $w(u_i, v_j) = 2r$ and $w(v_j, W) = 2r$, so $B_G(u_i, W) \ge 2r$.
From the $u_i$ column in Table~\ref{tab:lb_testing_winner_edge} we have that all weights of edges that enter the set $\{u_1, \ldots, u_r\}$ are at most $2r-2$
(recall our assumption that all entries in $C$ are positive).
Therefore, any path going from $W$ to $u_i$ must use such an entering edge, so $B_G(W, u_i) \le 2r-2 < B_G(u_i, W)$.
Thus, $W$ is not a winner.

Now we prove the reverse direction.
Suppose $C$ has no entry of value $r$.
First of all, $B_G(W, W') \ge 2r$ since $w(W, W') = 2r$.
To travel from $W'$ to $W$, the last edge will have a value corresponding to the column $W$ in Table~\ref{tab:lb_testing_winner_edge}, so the last edge has weight at most $2r$.
Therefore, $B_G(W', W) \le 2r \le B_G(W, W')$.

Secondly, for any $j \in [r]$, $B_G(W, v_j) \ge 2r$ since the path $W \rightarrow W' \rightarrow v_j$ has bottleneck $2r$.
Same as the previous case, in order to travel from $v_j$ to $W$, the last edge has weight at most $2r$, so $B_G(v_j, W) \le 2r \le B_G(W, v_j)$.

Finally, for any $i \in [r]$, $B_G(W, u_i) \ge 2r-2$ since $w(W, u_i) = 2r-2$.
To travel from $u_i$ to $W$, we need to leave the set $\{u_1, \ldots, u_r\}$ at some point.
When we do it, we use an important edge weight from the row $u_i$ of Table~\ref{tab:lb_testing_winner_edge}.
Since $C$ has no weight $r$ entry, all important weights in row $u_i$ of Table~\ref{tab:lb_testing_winner_edge} are at most $2r-4$.
Therefore, $B_G(u_i, W) \le 2r-4 < B_G(W, u_i)$.

Thus, $W$ is a winner when $C$ has no entry of value $r$.
\end{proof}

Hence, if we run the assumed $\Oh(T(n)) = \Oh(T(10r-2))$ time algorithm on the preference profile described above to check whether candidate $W$ is a Schulze winner, we could use Claim~\ref{cl:lb_testing_winner} to decide if $C$ has an entry of value $r$, and thus solve the Dominating Pairs problem.
If $T(n)$ is super polynomial, then the theorem trivially holds as the Dominating Pairs problem is polynomial-time solvable; otherwise we have $\Oh(T(10r-2)) = \Oh(T(r))$.

The same construction also works for the reduction from Dominating Pairs to \schulzew.

Suppose $C$ contains an entry $C_{i,j}=r$ for some $i, j \in [r]$. In this case, it is easy to check that $B_G(u_i, W) \ge 2r$, $B_G(u_i, W') \ge 2r$ and $B_G(u_i, v_{j'}) \ge 2r$ for every $j' \in [r]$. However, we have that all weights of edges that enter the set $\{u_1, \ldots, u_r\}$ are at most $2r-2$ (recall our assumption that all entries in $C$ are positive), which means that $B_G(W, u_i) \le 2r-2$, $B_G(W', u_i) \le 2r-2$ and $B_G(v_{j'}, u_i) \le 2r-2$ for every $j' \in [r]$. This means that only vertices in $\{u_1, \ldots, u_r\}$ can be Schulze winners when $C$ has an entry of value $r$.

Now suppose $C$ does not have an entry of value $r$. In this case, we see that $B_G(W, u_i) \ge 2r-2$ for every $i \in [r]$. However, in order to leave the set $\{u_1, \ldots, u_r\}$, we must use an edge of weight at most $2r-4$, so $B_G(u_i, W) \le 2r-4$ for every $i \in [r]$. Therefore, none of the vertices in $\{u_1, \ldots, u_r\}$ can be a Schulze winner when $C$ does not have an entry of value $r$.

Hence, if we run the assumed $\Oh(T(n)) = \Oh(T(10r-2))$ time algorithm on the preference profile described above and find a Schulze winner (one always exists~\cite{Schulze11}), we can decide if the Dominating Pairs problem has a solution by checking if the winner is from the set $\{u_1, \ldots, u_r\}$. Similar as before, if $T(n)$ is super polynomial, then the theorem trivially holds as the Dominating Pairs problem is polynomial-time solvable; otherwise we have $\Oh(T(10r-2)) = \Oh(T(r))$.
\end{proof}

\section{Conclusions}
This paper considered the Schulze voting method and gave new algorithms and fine-grained conditional lower bounds for central problems such as computing the weighted majority graph (useful for other voting rules as well), verifying whether a candidate is a winner, finding an arbitrary winner and computing all winners.

It is worth mentioning that while we focused on weighted majority graphs similarly to previous works\footnote{These assumed additionally that votes are linear orders---in contrast, we allow weak preference orders.} \cite{ParkesX12,GaspersKNW13,MentonS13,ReischRS14,HemaspaandraLM16,CsarLP18} our algorithms work for {\em arbitrary} weighted directed graphs; let us call these {\it comparison graphs}.
This means that we cover all possible weak orders $\succeq_D$ on $\naturals_0 \times \naturals_0$ that compare {\it the strength of the link} $(M(u,v),M(v,u)) \in \naturals_0 \times \naturals_0$~\cite{Schulze11}.
For a given instance with $n$ voters and $m$ candidates there are exactly $m(m-1)$ direct ordered comparisons between the candidates.
Even if $\succeq_D$ is defined on $(n+1)^2$ different pairs, only at most $m(m-1)$ pairs appear in the instance.
Therefore, for a given instance, we can encode the present pairs as numbers from the set $\{1,2,\dots,m(m-1)\}$ and use a standard comparison relation instead of $\succeq_D$.
Then we simply use these numbers as weights in the comparison graph.
In such a way, our algorithms work for all comparison relations mentioned by Schulze~\cite{Schulze11}, i.e., {\it margin}, {\it ratio}, {\it winning votes}, and {\it losing votes}.
Note that a weighted majority graph is defined as a comparison graph for the relation $\succeq_{\text{margin}}$ s.t. $(a,b) \succeq_{\text{margin}} (c,d)$ if and only if $a-b \geq c-d$.

Despite the fact that we have a lower bound on constructing the weighted majority graph it does not mean we have a lower bound on any voting rule which is defined using the weighted majority graph (e.g. tournament solutions~\cite{BrandtBH16}).
Indeed, this does not exclude an other way of finding a winner under a voting rule without construction of the weighted majority graph.
In the Schulze method it happens that the lower bound for finding a winner (Theorem~\ref{thm:lb_testing_winner}) is almost the same (up to technical details) as the lower bound for constructing the weighted majority graph (Theorem~\ref{thm:lb_graph}).
It is interesting to investigate for which voting rules (originally defined using the weighted majority graph) we can find winners in time faster than $\Oh(m^{2.5-\eps})$ for some $\eps>0$.
It would require to use different techniques than these presented in the paper.

On the other side, it is worth applying a fine-grained complexity approach on voting rules similar to the Schulze method, in particular, these satisfying the axioms listed in the Section~\ref{sec:intro}.
Despite those similarities, they might require different techniques than used in this paper.
More generally, it is interesting to research on fine-grained complexity for other computational social choice problems (not necessarily related to graph problems) that already have polynomial time algorithms (e.g., dynamic programming approach used for singled-peaked elections~\cite{BetzlerSU13}).

\section*{Acknowledgments}
We thank David Eppstein for alerting us to the Wikipedia article~\cite{wikiwidest}.
We would like to thank the anonymous reviewers of EC 2021 for their helpful comments.

Virginia Vassilevska Williams was supported by an NSF CAREER Award, NSF Grants CCF-1528078, CCF-1514339 and CCF-1909429, a BSF Grant BSF:2012338, a Google Research Fellowship and a Sloan Research Fellowship.
Yinzhan Xu was supported by NSF Grant CCF-1528078.
Krzysztof Sornat was partially supported by the National Science Centre, Poland (NCN; grant number 2018/28/T/ST6/00366) and the European Research Council (ERC) under the European Union’s Horizon 2020 research and innovation programme (grant agreement No 101002854).
\begin{figure}[!h]
  \center
  \includegraphics[width=80pt]{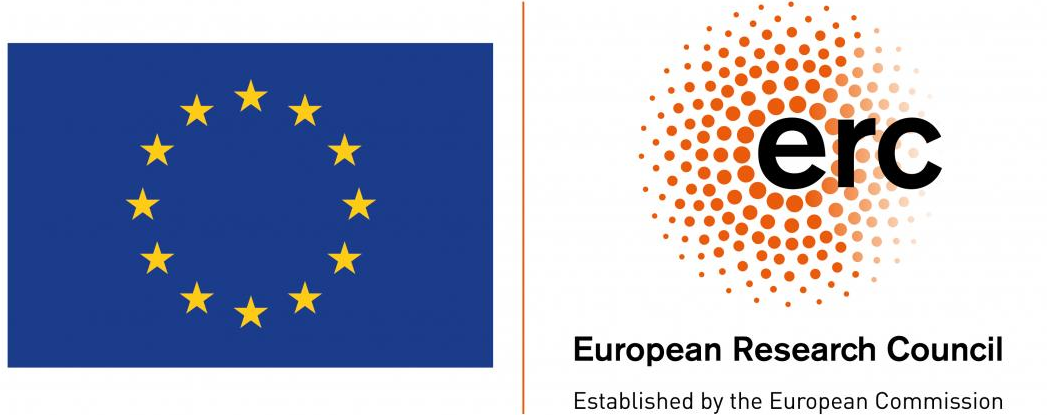}
\end{figure}

\bibliographystyle{alpha}
\bibliography{bib}

\end{document}